\documentclass{article}
\usepackage[preprint,nonatbib]{neurips_2025}

\usepackage[utf8]{inputenc} \usepackage[T1]{fontenc}    \usepackage{url}            \usepackage{booktabs}       \usepackage{amsfonts}       \usepackage{nicefrac}       \usepackage{microtype}      \usepackage{xcolor}         
\usepackage{algorithm}
\usepackage{algpseudocode}

\usepackage{graphicx}

\usepackage{amsmath, amsthm, amssymb}
\usepackage{hyperref}       
\newcommand{\E}{\mathop{\textrm{E}}}

\newcommand{\R}{\mathrm{I}\!\mathrm{R}}

\newtheorem{lemma}{Lemma}
\newtheorem{theorem}{Theorem}
\newtheorem{corollary}{Corollary}
\newtheorem{proposition}{Proposition}
\theoremstyle{definition}
\newtheorem{example}{Example}

\title{YEAST: Yet Another Sequential Test}

\author{Alexey Kurennoy, Majed Dodin, Tural Gurbanov, and Ana Peleteiro Ramallo\\
  }

\begin{document}

\maketitle

\begin{abstract}

Online A/B experiments require timely detection of harmful changes while controlling false alarms. We introduce YEAST, a simple sequential test for continuous monitoring of count and real-valued metrics over a prespecified event horizon. It compares the cumulative difference between experiment groups with a constant boundary determined by an estimate of the variance of this difference at the end of monitoring. Under the strong null of no effect on the event-generation process, we establish asymptotic Type-I error control through a maximal inequality and a central limit theorem. The required conditions allow some forms of repeated-user dependence. We evaluate YEAST using synthetic simulations and a semi-synthetic study of fixed-duration monitoring and observe power gains over the tested martingale-based methods.

\end{abstract}

\section{Introduction}

Online A/B experiments have become a standard approach for evaluating the impact of machine learning models on the systems and products where they are deployed. Detecting harmful effects promptly can protect user experience and reduce financial losses. Repeatedly applying a fixed-sample test as data accumulate can, however, inflate the false positive rate. Sequential tests account for repeated significance checks, allowing experiments to stop early while controlling false alarms.

Continuous monitoring allows significance checks as new observations arrive, including after every event. Both martingale methods and continuous-time extensions of group sequential testing support this flexibility~\cite{johari2022always,Howard2021,Bet:98}. Our focus is on a procedure that combines continuous monitoring with simple calibration and high sensitivity over a prespecified finite horizon.

We propose YEAST (YEt Another Sequential Test), which tracks the cumulative difference in a metric between the experiment groups and signals when that difference crosses a constant boundary. The metric can count events or sum real-valued outcomes, such as revenue. The boundary has a closed-form expression based on a prespecified maximum number of events and an estimate of the variance of the cumulative difference at that horizon.

Our theoretical contribution is an asymptotic maximal inequality that connects boundary-crossing probabilities to the terminal distribution. Together with a terminal central limit theorem and conditions on future increments, it yields asymptotic Type-I error control under the null that treatment leaves the event-generation process unchanged. This approach allows some repeated-user dependence without requiring a Brownian approximation for the whole trajectory. In the canonical Brownian setting, YEAST's one-sided boundary coincides with the continuous-monitoring O'Brien--Fleming boundary~\cite{BFl:79}.

We evaluate YEAST using synthetic simulations and a semi-synthetic study based on public online retail data. We observe power gains over the tested martingale-based methods, particularly at small and moderate effects.

To summarise, our \textbf{main contributions} are:

\begin{itemize}
    \item A sequential test for count and real-valued metrics, with a closed-form constant boundary that can be checked after every event up to a prespecified horizon.
    
    \item An asymptotic maximal bound inspired by the generalised L\'evy inequalities~\cite[Paragraph~29.1, A]{Loe:60}, together with Type-I error and power theorems.
    
    \item Synthetic and semi-synthetic evaluations of false detection rates, power, and sensitivity to variance and horizon calibration, with openly available code and simulation outputs~\cite{KGu:24}.
\end{itemize}

\noindent
The remainder of this paper is structured as follows. Section~\ref{sec:related} reviews related work. Section~\ref{sec:sst} introduces YEAST, explains its derivation, and presents its formal guarantees. Section~\ref{sec:validation} reports the semi-synthetic retail study. Section~\ref{sec:limitations} discusses limitations, and Section~\ref{sec:conclusion} concludes. Additional synthetic experiments appear in Appendix~\ref{sec:synthetic}.

\section{Related Work}\label{sec:related}

Our work is related to a large literature on sequential decision-making
based on sequences of hypothesis tests, confidence sequences, or sequences of ``always-valid'' p-values. While early work in this field was primarily concerned with clinical trials and industrial quality control~\cite{dodge1929method,bross1952sequential,armitage1954sequential}, more recent work is typically motivated by the sequential nature of online experiments, where users arrive sequentially and generate outcomes that are observed quickly relative to the duration of the experiment~\cite{larsen2023statistical, johari2022always, Howard2021}. Here, the goal of sequential testing is typically either to reduce the opportunity cost of longer experiments, as in the context of best-arm identification in multi-armed bandit problems~\cite{Bandit1,Bandit2,Bandit3,Bandit4}, or to manage risk by faster detection of harmful treatments~\cite{Lindon22, ham2023designbased}. 

\subsection{Martingale Methods and E-Processes}

Building on the seminal work of Wald~\cite{Wald45} and Doob \cite{doob1953stochastic}, most modern sequential procedures are based on maximal inequalities for appropriately constructed martingales that provide time-uniform statistical guarantees.
While Wald's original sequential probability ratio test (SPRT) was developed from scratch, its construction is essentially of this “martingale-type”. SPRT sequentially evaluates the likelihood ratio test statistic at each sample size for the null hypothesis $H_{0}: \theta = \theta_{0}$ against the simple alternative $H_{1}: \theta = \theta_{1}$. The test continues until the statistic crosses one of two predefined bounds, at which point the null hypothesis is either rejected or not rejected accordingly. SPRT requires specifying a single-point alternative, which limits its use for general-purpose experimental monitoring. To address this limitation, likelihood ratio methods using the method of mixtures (mSPRT) have been developed. These approaches replace the single-point alternative with a composite hypothesis and the standard likelihood ratio --- with a mixture, which remains a martingale under the null \cite{Robins1970}.

In addition to mSPRT, the class of martingale methods encompasses more sophisticated approaches constructed using other martingales and falling under the framework of generalised always-valid inference (GAVI) \cite{johari2022always,Howard2021,ham2023designbased,LAl:22,LHa:22,BKa:22}. 

Another recent branch of martingale-based methods is built upon the principle of testing by betting (see~\cite{PRa:23} and references therein).

The martingale-based approach was generalised through the concept of \emph{e-processes} \cite{Grunwald_deHeide_Koolen_2024_SafeTesting,Shafer_2021_TestingByBetting}~--- stochastic processes characterised by a property that allows for straightforward construction of sequential tests. The roots of this perspective can be traced back to \cite{Vovk_1993_LogicOfProbability}, who introduced the use of test martingales within a game-theoretic framework as a foundation for statistical testing. This remains an active area of research, with a growing literature developing new constructions and applications of e-processes \cite{Shekhar_Ramdas_2021_NonparametricTwoSampleByBetting_arxiv,Shekhar_Ramdas_2023_NonparametricTwoSampleByBetting_TIT}. See \cite{Ramdas_Wang_2025_HypothesisTestingWithEvalues} for a recent overview of the topic and further references.

\subsection{Group Sequential Methods}

A conceptually different approach is followed by methods based on group sequential tests (GST) \cite{Po:77,BFl:79,SWe:82,GDe:83}. GST methods rely on (asymptotic) characterisations of the joint distribution of a vector of test statistics across increments of data, which permits the computation of sequential boundaries that control the Type-I error rate across the scheduled analyses via multivariate (numerical) integration. 
Earlier GST methods (such as \cite{Po:77,BFl:79}) required a fixed number and timing of interim analyses. In contrast, the later generalisation \cite{GDe:83} permits flexible monitoring schedules without fixing the total number of analyses in advance. This flexibility is achieved using an alpha-spending function, which controls the allocation of Type-I error across analyses, based on a prespecified maximum information level.

Betensky~\cite{Bet:98} extends this perspective to continuous monitoring by constructing a continuous stopping boundary from a prescribed alpha-spending function. In the Brownian-motion formulation, the reflection principle gives an explicit crossing probability for a constant one-sided boundary~\cite[Section~2.2]{MPe:10}, providing the connection to our proposed method as discussed in Section~\ref{sec:sst}.

\subsection{Simple Sequential A/B Testing}
Our approach has similarities to \emph{simple sequential A/B testing} \cite{Mi:15}. However, the test from \cite{Mi:15} only allows monitoring count-type metrics (for example, the number of clicks or purchases) whilst our proposed method can be used to monitor metrics of any type, including real-valued financial metrics such as revenue. One can view our work as a generalisation of \cite{Mi:15} yet we emphasise that this generalisation is not at all trivial as the theory presented in \cite{Mi:15} fundamentally supports only count-type metrics.

\section{YEAST: YEt Another Sequential Test}\label{sec:sst}

In this section, we present a novel, alternative sequential testing method proposed in this paper.

\subsection{Problem Setup and Notation}\label{subsec:setup}

Suppose we are conducting an A/B test
and have a certain metric of interest.
We want to constantly track this metric and have a tool to check if the difference between the groups is at any point large enough to conclude that the treatment has an impact on the metric. The tool must have good control of the false detection rate despite the repeated checks. In other words, it must allow for peeking into the interim results of the experiment without inflating the false detection rate.

To develop a tool satisfying the above requirement, we consider the difference $S$ in the total value of the metric of interest accumulated by the control $(W = 0)$ and the treatment $(W = 1)$ groups. Typically, this difference $S$ can be represented as a sum of increments over a stream of events. For example, if the metric of interest is revenue then the tracked difference $S$ can be decomposed into the sum of order values over checkout events where the order value is taken with a positive sign if it comes from the control group and with a negative sign if it comes from the treatment arm. Mathematically, the increments can be written as 

\begin{equation}\label{eq:X=}
X_{i} = \left( 1 - 2 W_{i} \right) Y_{i},
\end{equation}
where $Y_i$ is the value or outcome associated with the $i$-th event (such as the order value in the example above) and $W_i$ is the event assignment indicator (it equals 0 if the event was generated by a subject from the control group and 1 if its subject belongs to the treatment group). With this notation at hand, the difference in the metric of interest after observing $n$ events can be written as the running sum
\begin{equation} \label{eq:Sn}
    S_n = \sum_{i=1}^n X_i,\quad n\ge 1,
\end{equation}
as illustrated in the example data set displayed in Table~\ref{tbl:Sn}.

\begin{table}[H]
\caption{An Example of Experiment Data}\label{tbl:Sn}
\centering
\begin{small}
\begin{tabular}{rrrrrr}
  \hline
 i & timestamp & group & Y & X & S \\ 
  \hline
1 & 2023-08-01 12:00:00 & control & 175.0 & 175.0 & 175.0 \\ 
  2 & 2023-08-01 12:00:02 & treatment & 35.5 & -35.5 &  139.5 \\ 
  3 & 2023-08-01 12:00:05 & treatment & 20.0 & -20.0 & 119.5 \\ 
  4 & 2023-08-01 12:00:10 & control & 100.0 & 100.0 & 219.5 \\ 
  \ldots & \ldots & \ldots & \ldots & \ldots & \ldots \\
   \hline
\end{tabular}
\end{small}
\end{table}

So far we have used revenue monitoring as an example but one can think of examples from other domains such as \emph{video streaming} (metric: hours streamed; the events are user sessions; the event value is the total duration of videos streamed within the session), \emph{online advertisement} (metric: number of clicks; the events are clicks; the event value is 1), \emph{medical drug trial} (metric: number of cases of side-effects; the events are incoming patients; the event value is the indicator of side-effect presence after taking the drug), \emph{online subscription services} (metric: number of new users who convert into a paid subscription; the events are newly acquired users; the event value is the indicator of subscribing), \emph{banking} (metric: total transaction fee; the events are transactions; the event value is fee amount), etc.

\subsection{Assumptions and the Null Hypothesis}\label{sec:assumptions}

We will now state and discuss our assumptions and the null hypothesis.

Let $U_i$ denote the subject generating event $i$, and let $A_u$ be the assignment sign of subject $u$, with $A_u=1$ for control and $A_u=-1$ for treatment. Thus $1-2W_i=A_{U_i}$ and $X_i=A_{U_i}Y_i$.

Throughout, the event stream denotes the sequence $\{(Y_i,U_i):i\ge1\}$ of event values and subject identities, excluding the assignment indicators $W_i$.

\textbf{Assumption 1 (Random Subject Assignment)} \textit{The subject signs $\{A_u\}$ are independent, with $\Pr\{A_u=1\}=\Pr\{A_u=-1\}=1/2$. Each subject retains the same assignment throughout the experiment.}

We work with the following strong null hypothesis.

\noindent\textbf{$\textup{H}_0$}: \textit{the treatment does not affect the event-generation process, including event values, subject identities, and event timing. In particular, the subject assignments are independent of the entire event stream:}
\begin{equation}\label{eq:H0a}
    \{A_u\}_u \perp\!\!\!\perp \{(Y_i,U_i):i\ge1\}.
\end{equation}
Together with Assumption~1, this implies
\begin{equation}\label{eq:H0b}
    \Pr\{W_i=1\mid (Y_j,U_j)_{j\ge1}\}=1/2\quad\textup{a.s.},\quad i\ge1.
\end{equation}

\noindent It is important to distinguish between subjects (such as users) and events generated by subjects (such as orders).
The event outcomes may be dependent and need not be identically distributed.
The event assignment indicators are identical for events from the same subject.
Under the strong null, the collection of subject assignments is independent of the entire event stream.
The event-level sequences $\{Y_i\}$ and $\{W_i\}$ can nevertheless be dependent on one another:
event values can provide information about whether two events come from the same subject and hence have identical assignment indicators.

The null hypothesis stated above exhibits the difference between our setup and the classical testing for a zero average treatment effect. Specifically,
\begin{itemize}
    \item our testing procedure acts on the level of events generated by experiment subjects and not the experiment subjects themselves (for example, orders and not customers);

    \item we assume that under the null hypothesis, the treatment does not affect the ``event generation process'' in any way.
\end{itemize}

In the revenue monitoring example, the latter implies that both the frequency of orders and the order value distribution stay the same in both groups under the null. We discuss the limitations of the considered setup in Section~\ref{sec:limitations}.

The approach described so far allows us to derive a simple but effective sequential testing procedure we present next.

\subsection{Proposed Sequential Testing Method}\label{subsec:sst-intro}

The proposed monitoring method consists of tracking the difference in the metric of interest between the experiment groups and (continuously) comparing it with a constant (alerting) boundary. Whenever the boundary is crossed, we reject the strong null at nominal significance level $\alpha$. The asymptotic guarantee requires the regularity conditions stated below.

The procedure requires setting the maximum number of events to be observed during the monitoring period that we denote by $N$. In other words, the boundary is computed assuming that the monitoring will not continue after $N$ events have been collected. See Section~\ref{sec:inputs} for how to set $N$ in practice.

The proposed method has only one other parameter that needs to be estimated beforehand - the (scaled) variance of the difference in the metric of interest at the end of monitoring, $V_N = var(S_N) / N$. The nonnegative estimator of $V_N$ needs to be relatively consistent under the null hypothesis, i.\,e. $\hat V_N/V_N\to_p1$. Given the estimate $\hat V_N$ (computed by such an estimator), the alerting boundary is set to
\begin{equation} \label{eq:constantThreshold}
    b^* = z_{1-\alpha/2}\cdot\sqrt{N\hat V_N},
\end{equation}
where $z_{1-\alpha/2}$ is the quantile of the standard normal distribution of level $1-\alpha/2$. The proposed simple sequential test proceeds by tracking $S_n$ and comparing it against $b^*$.

The necessity to fix or estimate $N$ and to provide an estimate of $V_N$ beforehand can be perceived as a disadvantage of the proposed approach. However, we argue that
in many practical scenarios, those parameters can be reliably set based on pre-experiment data (in fact, the underlying data distribution has to have some level of stability or regularity in time for the A/B testing to be meaningful in the first place).

\subsubsection{Informal Derivation of the Proposed Method}\label{subsubsec:informal} The derivation of the alerting boundary \eqref{eq:constantThreshold} follows two steps. First, under the regularity conditions in Section~\ref{subsec:statements}, an asymptotic inequality bounds the false detection rate by twice the crossing probability at the end of monitoring. For $b=\gamma\sqrt{var(S_N)}$ and fixed $\gamma>0$,
\begin{equation}\label{fdr-sst}
    FDR = \Pr\left\{\max_{n=1}^N S_n > b\right\} \le 2\Pr\{S_N>b\}+o(1).
\end{equation}
Here FDR denotes the probability of a false detection in one experiment.
In the second step we use the Central Limit Theorem to obtain
\begin{equation}\label{cross-prob-approx}
    \Pr\{S_N>b\}=1-\Phi\left(\frac{b}{\sqrt{NV_N}}\right)+o(1),
\end{equation}
where $\Phi$ denotes the CDF of the standard normal distribution. Setting $b=z_{1-\alpha/2}\sqrt{NV_N}$ gives $FDR\le\alpha+o(1)$. Theorem~\ref{thm:sst} establishes this oracle result, and Corollary~\ref{cor:plugin} justifies replacing $V_N$ by a (relatively) consistent estimate.

The proposed testing procedure is defined in Algorithm~\ref{alg:SST} and will be referred to as \emph{YEAST} (from YEt Another Sequential Test). The estimation of $V_N$ is discussed in Section~\ref{sec:inputs} below.

\begin{algorithm}
\caption{YEAST}\label{alg:SST}
\begin{small}
\begin{algorithmic}
\Require {$N$, $\hat V_N$, $\alpha\in(0,1)$}
\State $b^* \gets z_{1-\alpha/2}\sqrt{N\hat V_N}$
\For{$n=1,\,\ldots,\,N$}
    \If{$S_n > b^*$}
        \State \text{flag significance}
    \EndIf
\EndFor
\end{algorithmic}
\end{small}
\end{algorithm}

For a two-sided test, we track $|S_n|$ and compare it against
    $b^*_{2-sided} = z_{1-\alpha/4}\cdot\sqrt{N\hat V_N}$.

Mind that differently from the non-sequential testing, the significance level $\alpha$ needs to be divided by 2 for the one-sided test and by 4 for the two-sided test. The one-sided correction comes from the factor of two in bound~\eqref{fdr-sst}; the two-sided correction also accounts for the two possible directions of crossing. Corollary~\ref{cor:two-sided} gives the corresponding guarantee.

In the remainder of this subsection we draw a comparison between YEAST and GST methods~\cite{Po:77,BFl:79,GDe:83}, particularly focusing on their continuous-time variants~\cite{Bet:98}.

Continuous-time group sequential testing typically calibrates a stopping boundary using a Brownian-motion approximation to the test statistic, with time measured by accumulated statistical information. This approximation describes both the marginal distributions of the statistic and its dependence across monitoring times. Different boundary shapes can then be chosen to determine how readily the test rejects at different stages of the experiment~\cite{GDe:83,Bet:98}.

YEAST obtains an asymptotic bound on the probability of crossing its boundary using a terminal central limit theorem and conditions on the remaining increments. These conditions ensure that, conditional on the observed history, the remaining sum has probability at least approximately one-half of being nonnegative, except near the horizon, where its magnitude is negligible. This argument does not require a Brownian-motion approximation for the entire trajectory. It accommodates the repeated-customer dependence in our example, although its assumptions must still be verified for the event stream under consideration.

The resulting boundaries can coincide. In the canonical Brownian model, the reflection principle~\cite[Section~2.2]{MPe:10} gives the one-sided YEAST threshold $z_{1-\alpha/2}\sigma_N$, which is also the continuous-monitoring O'Brien--Fleming boundary on the cumulative-sum scale~\cite{BFl:79}. Here, $\sigma_N$ is the terminal standard deviation; YEAST permits its replacement by a consistent estimate. YEAST therefore provides an alternative justification for this simple boundary under its stated assumptions, including settings with dependent event contributions.

\subsubsection{Setting the Input Parameters}\label{sec:inputs}

As was described in the previous subsection, the computation of the alerting boundary of the proposed sequential testing procedure requires two inputs: the number of events $N$ to be observed during the monitoring process and an estimate $\hat V_N$ of the (scaled) variance of the tracked metric difference at the end of monitoring. In this subsection, we discuss how to provide those inputs in practice.

The formal results below use a deterministic event cap $N$ and require monitoring to stop by that cap. In practice, to set $N$, we estimate the expected number of events to be collected over the experiment time frame. In the simplest form, this can amount to setting $N$ to the number of events observed in a pre-experiment period of the same duration as the experiment we want to monitor.

For variance estimation, we suggest using pre-experiment data as well. Obtaining an accurate estimate this way requires the pre-experiment period to be representative of the experiment time. However, in cases where this cannot be assumed, the validity of fixed-duration randomised experiments is questionable in principle (as the effects measured during such experiments may not generalise beyond their time frames). So in the context of A/B testing, this is a natural assumption. Note that estimating variance from pre-experiment data is a standard procedure performed before an A/B test because it is used in computing the MDE (minimum detectable effect) and identifying the necessary sample size.

As was noted in Section~\ref{sec:assumptions}, the event outcomes $Y_i$, $i\ge 1$, can be dependent because some of the events are generated by the same user. The same holds for the event assignment indicators $W_i$, $i\ge 1$. Consequently, the increments $X_i$, $i\ge 1$, of the monitored trajectory can be serially correlated in which case the variance $var(S_N)$ does not equal the sum of variances $var(X_i)$. Under the strong null, contributions from different users are independent conditional on the event stream and are uncorrelated unconditionally.
Therefore, we propose to use an estimator that is robust to clustered serial correlation \cite{Cam:11} to estimate $V_N$. Implementations of such estimators are readily available, for example as part of the \texttt{sandwich} package in R \cite{ZSu:20, Zei:04} (see function \texttt{vcovCL}). The cluster variable in this case is the user identifier. As we will show in our experiments with real-world data in Section~\ref{sec:validation}, the use of a robust (clustered) variance estimator (instead of the simple sum of $var(X_i)$) can be utterly important for effective control of the false detection rate.

While the need to set a finite horizon $N$ can be viewed as a disadvantage of YEAST, we have not found it problematic in practice. When running A/B experiments it is typical to plan for a certain duration and the number of observations collected during the experiment time frame is usually predictable. In addition, the decision boundary of YEAST depends on the square root of $N$, which makes the method more robust to inaccuracies in the estimation of $N$. That said, we acknowledge that there are situations where it is difficult to reliably estimate $N$ and in that case, it may indeed be better to resort to other methods (e.g. \cite{Howard2021}) that do not require setting a finite horizon. In the empirical evaluations in Section~\ref{sec:validation}, we never assume that $N$ is known and estimate it from pre-experiment data (as part of the method). We then monitor throughout the planned calendar period without enforcing the estimated event count as a stopping cap, to assess practical robustness to uncertainty in $N$.

\subsubsection{Formal Statements}\label{subsec:statements}

In this subsection, we establish asymptotic guarantees for the false detection rate and power of the proposed method.
The variables $Y_i$, $W_i$, and $U_i$, $i\ge1$, are random variables on the probability space $(\Omega,\mathcal F,\Pr)$, and $\R$ denotes the real numbers.
All limits are as the deterministic event cap $N\to\infty$. The proofs are in Appendix~\ref{app:proofs}.

Let $S_0=0$ and consider the monitoring filtration and variance at the end of monitoring
\begin{equation}\label{eq:filtration}
    \mathcal F_n=\sigma\{(Y_i,W_i,U_i):1\le i\le n\},\qquad
    \sigma_N^2=var(S_N),\qquad V_N=\sigma_N^2/N.
\end{equation}
Set $M_N=\max_{1\le n\le N}S_n$, $b_N(\gamma)=\gamma\sigma_N$, and
\begin{equation}\label{eq:crossing-time}
    T_N(b)=\inf\{1\le n\le N:S_n>b\},
\end{equation}
with $T_N(b)=\infty$ if the boundary is never crossed. We use strict crossings throughout, as in Algorithm~\ref{alg:SST}.

\noindent\textbf{Assumption 2 (Terminal CLT)}
\textit{Let $Y_i$, $i\ge 1$, be square-integrable, $\sigma_N^2>0$ for all $N\ge 1$, and}
\begin{equation}\label{eq:terminal-clt}
    (S_N - \E[S_N])/\sigma_N\stackrel{d}{\longrightarrow}\mathcal N(0,1).
\end{equation}

\noindent\textbf{Assumption 3 (Future Increments)} \textit{There exist deterministic integers $1\le L_N\le N$, $N\ge 1$, such that $L_N\to\infty$, $L_N=o(N)$, and the following conditions hold.}
\begin{enumerate}
    \item[(a)] \textbf{Approximate Future Sign-Balance.} \textit{There exist numbers $\delta_N\downarrow0$ such that for every $k=1,\,\ldots,\,(N-L_N)$, we have}
    \begin{equation}\label{eq:sign-balance}
        \Pr\{S_N-S_k\ge0\mid\mathcal F_k\}\ge\frac12-\delta_N
        \quad\textup{a.s.}
    \end{equation}
    \item[(b)] \textbf{Negligible Late Increments.} \textit{For $H_N$ defined as}
    \begin{equation*}
        H_N=\max_{N-L_N<n\le N}|S_N-S_n|,
    \end{equation*}
    \textit{it holds that}
    \begin{equation}\label{eq:late-window}
        H_N/\sigma_N\stackrel{p}{\to}0.
    \end{equation}
\end{enumerate}
The sequence $L_N$ separates the early and middle parts of monitoring from its final window.
For repeated-user streams, part~(a) holds when the future contribution of already-observed users is negligible relative to the still-random contribution of unseen users and the latter satisfies a suitable conditional anti-concentration bound (see Proposition~\ref{prop:balance} below).

Part~(b) holds, for example, if $\sup_i\E|Y_i|<\infty$ and $L_N=o(\sigma_N)$ since $H_N\le\sum_{i=N-L_N+1}^N|Y_i|$.

\begin{lemma}[Approximate Terminal Reduction]\label{lem:levy}
    Under Assumptions~2--3, for every fixed $\gamma>0$,
    \begin{equation}\label{eq:levy1}
        \Pr\{M_N>\gamma\sigma_N\}\le2\Pr\{S_N>\gamma\sigma_N\}+o(1).
    \end{equation}
\end{lemma}

\begin{theorem}[False Detection Rate Control]\label{thm:sst}
    Suppose the strong null \eqref{eq:H0a}--\eqref{eq:H0b} and Assumptions~1--3 hold.
    
    Then for every fixed $\gamma>0$,
    \begin{equation}\label{eq:bound-with-gamma}
        \Pr\left\{\max_{n=1}^N S_n>\gamma\sqrt{NV_N}\right\}
        \le2(1-\Phi(\gamma))+o(1).
    \end{equation}
    In particular, for any fixed $\alpha\in(0,1)$, the oracle boundary $z_{1-\alpha/2}\sqrt{NV_N}$ has false detection probability at most $\alpha+o(1)$.
\end{theorem}

\begin{corollary}[Plug-in YEAST Boundary]\label{cor:plugin}
    Under the conditions of Theorem~\ref{thm:sst}, suppose the nonnegative estimate $\hat V_N$ satisfies
    \begin{equation}\label{eq:variance-consistency}
        \frac{N\hat V_N}{\sigma_N^2}\to_p1.
    \end{equation}

    Then, for fixed $\alpha\in(0,1)$,
    \begin{equation}\label{eq:plugin-control}
        \Pr\left\{M_N>z_{1-\alpha/2}\sqrt{N\hat V_N}\right\}\le\alpha+o(1).
    \end{equation}
\end{corollary}

\begin{corollary}[Two-Sided YEAST Boundary]\label{cor:two-sided}
    Suppose the conditions of Theorem~\ref{thm:sst} hold and \eqref{eq:sign-balance} also holds with $S_N-S_k\le0$ in place of $S_N-S_k\ge0$. 
    
    Then, for fixed $\gamma>0$,
    \begin{equation}\label{eq:levy2}
        \Pr\left\{\max_{n\le N}|S_n|>\gamma\sigma_N\right\}
        \le2\Pr\{|S_N|>\gamma\sigma_N\}+o(1)
        =4(1-\Phi(\gamma))+o(1).
    \end{equation}
    Consequently, the two-sided oracle boundary $z_{1-\alpha/4}\sigma_N$ controls the false detection probability at $\alpha+o(1)$. Under \eqref{eq:variance-consistency}, the same holds for $z_{1-\alpha/4}\sqrt{N\hat V_N}$.
\end{corollary}

To interpret Assumption~3(a), let $\mathcal U_k=\{U_i:1\le i\le k\}$ and decompose the remainder as $S_N-S_k=B_{k,N}+R_{k,N}$, where
\begin{align}
    B_{k,N}&=\sum_{i=k+1}^N A_{U_i}Y_i\,\mathbf1\{U_i\in\mathcal U_k\},\label{eq:seen}\\
    R_{k,N}&=\sum_{i=k+1}^N A_{U_i}Y_i\,\mathbf1\{U_i\notin\mathcal U_k\}.\label{eq:unseen}
\end{align}
Let $\mathcal G_{k,N}=\mathcal F_k\vee\sigma\{(Y_i,U_i):1\le i\le N\}$. Conditional on $\mathcal G_{k,N}$, $B_{k,N}$ is known, while unseen users still have independent fair assignment signs. Therefore $R_{k,N}$ is conditionally symmetric, with conditional variance
\begin{equation}\label{eq:unseen-variance}
    s_{k,N}^2=\sum_{u\notin\mathcal U_k}\left(\sum_{i=k+1}^N Y_i\mathbf1\{U_i=u\}\right)^2.
\end{equation}

\begin{proposition}[A Sufficient Condition for Sign-Balance]\label{prop:balance}
    Suppose Assumption~1 and the strong null hold. Let $L_N\to\infty$ be deterministic integers with $1\le L_N\le N$ and $L_N=o(N)$. Suppose that for some deterministic $\epsilon_N\downarrow0$ and $\rho_N\downarrow0$, almost surely, simultaneously for every $1\le k\le N-L_N$,
    \begin{equation}\label{eq:balance-sufficient}
        s_{k,N}>0,\qquad |B_{k,N}|\le\epsilon_Ns_{k,N},\qquad
        \Pr\{|R_{k,N}|\le\epsilon_Ns_{k,N}\mid\mathcal G_{k,N}\}\le\rho_N.
    \end{equation}
    Then Assumption~3(a) holds with $\delta_N=\rho_N$, for both signs of the running sum.
\end{proposition}
The last condition in \eqref{eq:balance-sufficient} prevents the still-random sum from concentrating in an interval as small as the already-revealed contribution.
For example, suppose the conditional CDF of $R_{k,N}/s_{k,N}$ given $\mathcal G_{k,N}$ differs from $\Phi$ by at most a deterministic $r_N\downarrow0$, almost surely, uniformly over its argument and $1\le k\le N-L_N$. Then
\begin{align}\label{eq:balance-sufficient-proof}
    \Pr\{|R_{k,N}|\le\epsilon_Ns_{k,N}\mid\mathcal G_{k,N}\}
    &\le\Phi(\epsilon_N)-\Phi(-\epsilon_N)+2r_N\nonumber\\
    &\le\frac{2\epsilon_N}{\sqrt{2\pi}}+2r_N.
\end{align}
The first inequality uses the conditional CDF at $\epsilon_N$ and its left limit at $-\epsilon_N$, so it also covers atoms. The second follows by integrating the standard normal density $\phi(x)=(2\pi)^{-1/2}e^{-x^2/2}\le(2\pi)^{-1/2}$ over an interval of length $2\epsilon_N$. Thus the last condition in \eqref{eq:balance-sufficient} holds with $\rho_N=2\epsilon_N/\sqrt{2\pi}+2r_N$.

\begin{example}[Daily Cohorts with Repeat Purchases]
Suppose that, starting with no customers, each day brings an independent, identically distributed number of first-time buyers in $\{0,\ldots,J\}$, with positive mean and fixed integer $J\ge1$. Each customer makes an initial purchase and, independently on each of the following $D$ days, another purchase with probability $p\in(0,1)$, for a fixed integer $D\ge1$. No further purchases occur after that window. Order values lie in $[a,b]$, where $0<a<b<\infty$. Customer purchase histories, measured relative to the first-purchase day, are independent and identically distributed across customers and independent of the daily cohort sizes. Purchases are ordered chronologically, with any ties resolved independently of assignments. Each customer receives an independent fair assignment sign, retained for all their purchases, and the entire purchase process is independent of these signs. Thus Assumption~1 and the strong null hold, while customers can make several purchases on different days, interleaved with those of other customers. For example, $D=30$ allows repeat purchases over a month.

We first verify Assumption~2. Let $Q_u$ and $T_u$ be customer $u$'s total number of purchases and total spending over their full purchase history, and let $T_{u,N}=\sum_{i=1}^N Y_i\mathbf1\{U_i=u\}$ be their spending through event $N$. Let $\mathcal H=\sigma\{(Y_i,U_i):i\ge1\}$ and $q_N=var(S_N\mid\mathcal H)=\sum_u T_{u,N}^2$. After the $N$-th event, at most $J(D+1)$ already-observed customers have purchases remaining beyond the observation horizon. Completing these histories and including any customers whose first purchase occurs later on the same day changes the purchase count and the sum of squared customer totals by bounded amounts. The law of large numbers for the independent daily cohorts therefore gives
\[
    \frac{q_N}{N}\to v:=\frac{\E[T_u^2]}{\E[Q_u]}>0
    \quad\textup{a.s.}
\]
Since $a^2\le q_N/N\le(D+1)b^2$ and $\E[S_N\mid\mathcal H]=0$, bounded convergence also gives $\sigma_N^2/N=\E[q_N]/N\to v$. Conditional on $\mathcal H$, the sum has independent signed customer contributions bounded by $(D+1)b$, with total variance at least $a^2N$. Their conditional CLT, together with $q_N/\sigma_N^2\to1$ almost surely, yields Assumption~2.

Next, to verify Assumption~3, choose $L_N=\lfloor N^{1/3}\rfloor$, so $L_N\to\infty$ and $L_N=o(\sqrt N)$. For part~(b), simply note that each increment has absolute value at most $b$, so $H_N\le bL_N$. Since $\sigma_N^2/N\to v>0$, the chosen $L_N$ gives $H_N/\sigma_N\le bL_N/\sigma_N\to0$. It remains to verify part~(a).

For all sufficiently large $N$ such that $L_N>2M$, where $M=JD(D+1)$, and each $1\le k\le N-L_N$, consider the remainder $S_N-S_k$ of the running sum. Among already-seen customers only those whose first purchase occurred within the last $D+1$ calendar days can contribute to this remainder. There are at most $J(D+1)$ such customers, each with at most $D$ purchases remaining. We therefore have that $|B_{k,N}|\le Mb$ and at least $N-k-M$ future purchases belong to unseen customers. For each unseen customer $u$, the total $T_{u,N}$ includes only events after $k$ and satisfies $0\le T_{u,N}\le T_u\le(D+1)b$. Conditional on $\mathcal G_{k,N}$, the contributions of unseen customers are independent fair signs multiplied by these totals, and
\[
    s_{k,N}^2=\sum_{u\notin\mathcal U_k} T_{u,N}^2\ge a^2(N-k-M),\qquad
    \frac{\sum_{u\notin\mathcal U_k} T_{u,N}^3}{s_{k,N}^3}\le\frac{(D+1)b}{s_{k,N}}.
\]
For any $1\le k\le N-L_N$, the first inequality gives $s_{k,N}^2\ge a^2(L_N-M)\ge a^2L_N/2$ for all sufficiently large $N$, since $M$ is fixed. Let $F_{k,N}$ be a regular conditional CDF of $R_{k,N}/s_{k,N}$ given $\mathcal G_{k,N}$. The Berry--Esseen bound gives deterministic constants $C_0<\infty$ and $N_0$ such that, on a single event of probability one, for every $N\ge N_0$,
\[
    \max_{1\le k\le N-L_N}\sup_{x\in\R}
    |F_{k,N}(x)-\Phi(x)|\le\frac{C_0}{\sqrt{L_N}}.
\]
Set the deterministic quantity $d_N=\sqrt2\,Mb/(a\sqrt{L_N})$. Since $Mb/s_{k,N}\le d_N$, applying the conditional CDF bound at the endpoints $\pm d_N$ and bounding the normal density as in \eqref{eq:balance-sufficient-proof} gives, on the same event, for every $N\ge N_0$ and all $1\le k\le N-L_N$,
\begin{align*}
    \Pr\{|R_{k,N}|\le Mb\mid\mathcal G_{k,N}\}
    &\le F_{k,N}(d_N)-F_{k,N}(-d_N-)\\
    &\le\Phi(d_N)-\Phi(-d_N)+\frac{2C_0}{\sqrt{L_N}}\\
    &\le\frac{2d_N}{\sqrt{2\pi}}+\frac{2C_0}{\sqrt{L_N}}
    =\frac{C}{\sqrt{L_N}},
\end{align*}
where $C=2Mb/(a\sqrt\pi)+2C_0$ is deterministic. The left limit in the above formula includes any mass at the lower endpoint, and the conditional probabilities use the same regular conditional distributions.
Conditional symmetry and $|B_{k,N}|\le Mb$, followed by conditioning on $\mathcal F_k$, prove Assumption~3(a).

This example relies on bounded daily inflow, bounded order values, and a finite purchase window per customer; it permits repeat purchases across days or weeks within that window.
\end{example}

A bounded number of events per user alone is insufficient for the sign-balance assumption. To see this, set $Y_i=1$ and construct successive blocks indexed by $j\ge1$. Block $j$ contains one event from each of $m_j=2^j$ new users, followed by a second event from each of those users. Let $N_j=2\sum_{\ell=1}^j m_\ell$ be the end of block $j$ and $k_j=N_j-m_j$ its midpoint. Then $S_{N_j}-S_{k_j}$ equals the sum of the $m_j$ signs revealed in the first half of that block. On the event that this sum is negative, the conditional probability of a nonnegative remainder is zero; that event has probability tending to $1/2$. Since $m_j>N_j/4$, every sequence $L_N=o(N)$ has $k_j\le N_j-L_{N_j}$ for all sufficiently large $j$. Thus Assumption~3(a) fails although each user contributes only two events.

Theorem~\ref{thm:sst}, in particular, describes sensitivity to variance calibration. If the boundary uses $rNV_N$ in place of $NV_N$, for some fixed $r>0$, the asymptotic false detection bound is $2(1-\Phi(z_{1-\alpha/2}\sqrt r))+o(1)$. For example, overestimating the variance by 20\% gives an upper bound of approximately 3\% instead of 5\%.

Our second theorem shows that, under a fixed positive mean increment and the stated variance conditions, the detection probability tends to one as the monitoring horizon grows.

\begin{theorem}[Power]\label{thm:power}
    Let $Y_i\colon\Omega\mapsto\R$ and $W_i\colon\Omega\mapsto\{0,\,1\}$, $i\ge1$, be random variables, and let $S_n$, $n\ge 1$, be the running sum defined by \eqref{eq:X=}--\eqref{eq:Sn}. Suppose that Assumption~1 holds and $(S_N-N\mu)/\sqrt{var(S_N)}\stackrel{d}{\to}\mathcal N(0,1)$. Assume that $var(S_N) > 0$ for all $N\ge 1$ and that the deterministic sequence $V_N = var(S_N)/N$, $N\ge 1$, is bounded above. Finally, let $\E[X_i] = \mu > 0$ for every $i\ge1$ (implying that the null hypothesis \eqref{eq:H0a}--\eqref{eq:H0b} is violated), and let $\hat V_N$ be a nonnegative variance estimate satisfying $\hat V_N=O_p(1)$ under this alternative.
    Then, for any fixed $\alpha\in(0,1)$,
$
        \Pr\left\{\max_{n=1}^N S_n > z_{1-\alpha/2}\sqrt{N\hat V_N}\right\} \to 1$ as $N\to\infty$.
\end{theorem}

\section{Validation Using Real-World Data: a Semi-Synthetic Experiment}\label{sec:validation}

The best possible real-world validation would be using data from a very large number of A/B tests (some with statistically significant results and some with neutral outcomes under high power). Having those data, one could apply the compared sequential tests retrospectively and measure their false detection rates and sensitivity (power). As we do not have access to a sufficient number of past A/B tests, we chose a different approach and conducted a semi-synthetic study, in which we took 
a public real-world data set (``Online Retail'' dataset \cite{Che:15}), randomly assigned users in that data set to control and test, computed the cumulative difference in the metric of interest (revenue) between control and treatment for each assignment replication, and measured the associated detection rates. (For power assessment, we also artificially decreased the order values considering multiple effect sizes, see Section~\ref{subsec:power} below.)

The code implementing our experiments and the saved simulation outputs are openly available in the companion repository \cite{KGu:24}.

Here are the methods we compared.

\noindent\textbf{YEAST} The proposed sequential method using the significance boundary \eqref{eq:constantThreshold} as presented in Section~\ref{sec:sst}.

\noindent\textbf{mSPRT} The mixture sequential probability ratio test \cite{Robins1970,Lindon22}. We set the tuning parameter $\phi$ to $1/0.3^2\approx11.11$, 25, and 100 and denote the corresponding versions as mSPRTphi11, mSPRTphi25, and mSPRTphi100.

\noindent\textbf{GAVI} The generalisation of the always valid inference, as proposed in~\cite{Howard2021}. As in~\cite{SSe:23}, we set the numerator of parameter $\rho$ of the method to 250, 500, and 750 and denote the corresponding instances of the method by GAVI250, GAVI500, and GAVI750, respectively. We also included GAVI10K, with the numerator of $\rho$ set to 10,000, following the Eppo configuration reported in~\cite{SSe:23}.

\noindent\textbf{LanDeMetsOBF} The GST method \cite{GDe:83} with the O'Brien-Fleming alpha-spending function \cite{BFl:79}. For computational reasons, we constructed the boundary using 100 interim checkpoints and held the standardised critical values constant between checkpoints. This interpolation permits checking after every observation but does not inherit the discrete-look Type-I error guarantee.

\noindent\textbf{SeqC2ST-QDA} The sequential predictive test from \cite{PRa:23} with the online Newton step (ONS) strategy for selecting betting fractions. We used QDA (quadratic discriminant analysis) as a classification model for this method.

\subsection{Type-I Error (False Detection Rate)}\label{subsec:fdr}

The public dataset \cite{Che:15} consists of all transactions that occurred between 2010-12-01 and 2011-12-09 on a UK-based online store. We used the data up until 2011-11-30 to have 12 complete months and dropped orders with missing customer identifiers. The resulting dataset had 21,269 orders. We split it into two halves: the first 6 months were used for estimating parameters ($N$ and $V_N$) and the latter 6 months - for validation. Using the validation period, we randomly assigned customers to control and treatment 100,000 times\footnote{
We generated the assignments in parallel using 13 processes. The random seed of the $i$-th process was set to $i$. Each process but the last made 7692 replications while the last one made 7696 so the total number of replications was 100,000. The code and saved per-process results are available in the companion repository.}, computed the cumulative difference in the revenue between control and treatment for each assignment replication, applied the sequential tests to monitor the cumulative difference, and measured the frequency of detections over the replications.
Since no actual treatment was applied to the replicated treatment groups, the null hypothesis was true, and the computed rates were false detection rates.

It is typical of online transaction data to contain outliers (e.g. large wholesale customers that can easily skew the monitored metric difference in favour of one of the groups). Therefore, before replicating the assignments, we filtered events using a cap at the 99.9th percentile of total customer revenue in the pre-experiment period. Within each period, we retained an event only when that customer's cumulative revenue through the event was at most the cap. The same cap was applied before assignment replication and before introducing treatment effects.

The measured detection rates are reported in Table~\ref{tbl:real-world-experiment-online-retail} together with normal-approximation confidence intervals, reported as estimate ± margin of error and using a Bonferroni correction for nominal simultaneous 95\% coverage across the method--effect cells in the retail evaluation. Detection rates within one percentage point of the nominal significance level (set to 5\%) are shown in bold.

Following the variance discussion in Section~\ref{sec:inputs}, we employed the sequential tests with two different variance estimates. One of the two assumed that the increments $X_i$ of the monitored difference-in-sum are independent. The other was an estimator robust to cluster serial correlation, that is - to the correlation between increments coming from the same user in our case.\footnote{We used the \texttt{vcovCL} function of the \texttt{sandwich} R package to compute the robust variance estimate. To reduce the sparsity in the time dimension, we summed the data (revenue increments) by user-month before computing the estimator. See the linked repository \cite{KGu:24} for details.} The results corresponding to the use of the two estimators are reported in columns `non-robust' and `robust' of Table~\ref{tbl:real-world-experiment-online-retail}, respectively.\footnote{SeqC2ST-QDA does not take an estimate of the increment variance as input hence there is only one result for this method.}

Two observations from our experiments are in order. Firstly, YEAST with the robust variance estimator was the only method whose empirical false detection rate was within one percentage point of the nominal level without exceeding it. Several benchmarks also stayed below the nominal level but were more conservative. Secondly, using the non-robust variance estimate substantially increased false detections. This highlights the importance of accounting for repeated-user dependence when estimating the variance.

To assess realistic fixed-duration use, we deliberately used the pre-experiment estimate of $N$ only to calibrate the boundary and continued monitoring throughout the six-month validation period, without enforcing that estimate as a stopping cap. The period contained 12,068 retained events, compared with the estimate of 8,972. The reported results therefore assess empirical robustness to horizon misspecification in this mode of use.\footnote{As a supplementary comparison, we measured a false detection rate of 0.022 when monitoring stopped at the estimated cap, compared with 0.044 for the full period. At a 20\% treatment effect, the corresponding detection rates were 0.446 and 0.676.}

\begin{table}[ht]
    \caption{Semi-synthetic Experiment: Empirical Type-I Error (Online Retail).}\label{tbl:real-world-experiment-online-retail}
    \centering
    \begin{small}
    \setlength{\tabcolsep}{4pt}
\begin{tabular}{llcc}
  \hline
 & method & robust var. est. & non-robust var. est.\\
  \hline
  1 & YEAST & \textbf{0.044 ± 0.002} & 0.168 ± 0.004 \\
  2 & mSPRTphi100 & 0.002 ± 0.000 & 0.097 ± 0.003 \\
  3 & mSPRTphi11 & 0.000 ± 0.000 & 0.057 ± 0.003 \\
  4 & mSPRTphi25 & 0.001 ± 0.000 & 0.073 ± 0.003 \\
  5 & GAVI250 & 0.001 ± 0.000 & 0.082 ± 0.003 \\
  6 & GAVI500 & 0.001 ± 0.000 & 0.094 ± 0.003 \\
  7 & GAVI750 & 0.002 ± 0.000 & 0.098 ± 0.003 \\
  8 & GAVI10K & 0.004 ± 0.001 & 0.074 ± 0.003 \\
  9 & LanDeMetsOBF & 0.027 ± 0.002 & 0.133 ± 0.004 \\
  10 & SeqC2ST-QDA & \multicolumn{2}{c}{0.242 ± 0.005} \\
   \hline
\end{tabular}
\end{small}
\end{table}

\subsection{Type-II Error (Power) Assessment}\label{subsec:power}

To assess power, we conducted the same simulations as in Section~\ref{subsec:fdr} but additionally decreased the order values belonging to the treatment group in each replication. The relative decrease was equal to 5, 10, 20, and 50\%. 

To evaluate sensitivity, we included a non-sequential normal-approximation test, applied once at the end of each replication using the full validation dataset and the pre-experiment cluster-robust variance estimate. This is the one-sided fixed-horizon $z$-test implemented as \texttt{Classical} in the companion code. It serves as a reference without an adjustment for repeated significance checks.

The reported results use the robust variance estimate and appear in Table~\ref{tbl:semi-synthetic-power}. In this full-period retrospective evaluation, YEAST had power close to the fixed-horizon reference. SeqC2ST-QDA is shown in italics because its inflated null detection rate prevents a comparison at the nominal level. Bold entries have confidence intervals overlapping that of the highest-power eligible method, where eligibility excludes methods whose Bonferroni-adjusted null interval lies entirely above 5\%. 

\begin{table}[ht]
    \caption{Semi-synthetic Experiment: Power (Online Retail).}\label{tbl:semi-synthetic-power}
    \centering
    \begin{small}
    \setlength{\tabcolsep}{4pt}
\begin{tabular}{rlllll}
  \hline
 & method & \multicolumn{4}{c}{relative decrease in the order value}\\
 &  & 0.05 & 0.1 & 0.2 & 0.5\\
  \hline
 & \textit{Non-seq. $z$-test} & \textit{0.114 ± 0.004} & \textit{0.255 ± 0.005} & \textit{0.685 ± 0.005} & \textit{1.000 ± 0.000} \\
  1 & YEAST & \textbf{0.114 ± 0.004} & \textbf{0.252 ± 0.005} & \textbf{0.676 ± 0.005} & \textbf{1.000 ± 0.000} \\
  2 & mSPRTphi100 & 0.006 ± 0.001 & 0.019 ± 0.002 & 0.140 ± 0.004 & 0.998 ± 0.000 \\
  3 & mSPRTphi11 & 0.001 ± 0.000 & 0.007 ± 0.001 & 0.071 ± 0.003 & 0.993 ± 0.001 \\
  4 & mSPRTphi25 & 0.002 ± 0.001 & 0.010 ± 0.001 & 0.093 ± 0.003 & 0.996 ± 0.001 \\
  5 & GAVI250 & 0.003 ± 0.001 & 0.013 ± 0.001 & 0.106 ± 0.003 & 0.997 ± 0.001 \\
  6 & GAVI500 & 0.005 ± 0.001 & 0.017 ± 0.001 & 0.130 ± 0.004 & 0.998 ± 0.000 \\
  7 & GAVI750 & 0.006 ± 0.001 & 0.020 ± 0.002 & 0.146 ± 0.004 & 0.999 ± 0.000 \\
  8 & GAVI10K & 0.012 ± 0.001 & 0.038 ± 0.002 & 0.230 ± 0.005 & 1.000 ± 0.000 \\
  9 & LanDeMetsOBF & 0.075 ± 0.003 & 0.181 ± 0.004 & 0.585 ± 0.005 & \textbf{1.000 ± 0.000} \\
  10 & SeqC2ST-QDA & \textit{0.237 ± 0.005} & \textit{0.253 ± 0.005} & \textit{0.366 ± 0.005} & \textit{0.964 ± 0.002} \\
   \hline
\end{tabular}
\end{small}
\end{table}

\section{Limitations}\label{sec:limitations}

The empirical comparisons in Section~\ref{sec:validation} show favourable calibration and power for YEAST despite its simplicity. Yet there are scenarios where the underlying null hypothesis of no effect on the data generation process is too strong, potentially leading to some ``undesirable sensitivity''. In other words, sometimes the treatment can alter the data generation process but have a neutral average effect on the metric of interest. We can think of (at least) two such scenarios.

\textit{Situations where the treatment brings multiple changes to the data generation process and those changes counter-balance each other leading to a neutral average effect on the metric.} For example, the treatment can make users place fewer orders but increase the average order value at the same time so the two effects neutralise each other and the revenue per user stays intact.

\textit{Situations where the treatment changes the variance of the event outcome, leaving the average unaffected.} For example, the treatment can increase the variability in the order values without changing their average.

If in a given practical application, having sensitivity to one of the scenarios above is highly undesirable, it is safer to adhere to alternative sequential testing methods (e.g. \cite{Howard2021,ham2023designbased}).

Note that none of the problematic two scenarios can materialise when the event outcome is constant (i.\,e. the metric of interest simply counts events). In addition, the first scenario cannot hold when the treatment cannot affect the number of events, for example, when the events correspond to incoming patients in a medical study.

We see studying the behaviour of our proposed method in situations mentioned in this section and seeking mitigation strategies as a topic for future research.

It is worth emphasising that in some situations it can be desirable to capture all relevant potential defects (beyond the effect on the mean)~\cite{Lindon22}. For example, an increased outcome variance can indicate that the effect of the treatment is heterogeneous and some important subgroups of the population can be negatively affected (even though the overall mean stays the same). In such scenarios, the extra sensitivity mentioned in this section is not really a limitation but an advantage instead.

\section{Conclusion}\label{sec:conclusion}

This paper proposes a novel sequential test for continuous experiment monitoring. The proposed test supports both discrete and real-valued metrics. Its effectiveness is demonstrated in an empirical study (a semi-synthetic experiment based on real-world data).
We hope that the proposed method will enable a wider use of sequential testing for online evaluation.

\begin{ack}
The authors are very grateful to Evan Miller for insightful comments and his excellent blog \cite{Mi:15}, which sparked the idea for this research.

This research was carried out while all the authors worked at Zalando. The paper does not represent, imply, or establish any form of partnership, joint venture, or collaboration --- whether formal or informal --- between the organisations with which the authors are affiliated at present.
\end{ack}

\section*{Appendix}

\appendix

\section{Proofs and Auxiliary Statements}\label{app:proofs}

\begin{proof}[Proof of Lemma~\ref{lem:levy}.]
Fix $\gamma>0$, write $b_N=\gamma\sigma_N$, $m=N-L_N$, and $T=T_N(b_N)$. Since $\{S_N>b_N\}\subseteq\{T<\infty\}$,
\begin{equation}\label{eq:crossing-decomposition}
    \Pr\{T<\infty\}=\Pr\{S_N>b_N\}+\Pr\{T<\infty,S_N\le b_N\}.
\end{equation}
We bound disappearing crossings separately according to whether they occur by time $m$ or later.

For early crossings, let $\tau=T\wedge m$ and $E=\{T\le m\}$. This is a bounded stopping time, $E\in\mathcal F_\tau$, and on $E$ we have $S_\tau>b_N$. Since $\{\tau=k\}\in\mathcal F_k$ and $1\le\tau\le m$, partitioning over these events gives
\begin{align*}
    \Pr\{S_N-S_\tau\ge0\mid\mathcal F_\tau\}
    &=\sum_{k=1}^{m}\mathbf1\{\tau=k\}
      \Pr\{S_N-S_k\ge0\mid\mathcal F_k\}\\
    &\ge\frac12-\delta_N\quad\textup{a.s.}
\end{align*}
The inequality follows from Assumption~3(a), whose deterministic-time bounds hold simultaneously for the finitely many $k\le m$. Consequently, by conditional expectation,
\begin{align*}
    \Pr\{E,S_N>b_N\}
    &\ge\E\left[\mathbf1_E\Pr\{S_N-S_\tau\ge0\mid\mathcal F_\tau\}\right]\\
    &\ge(1/2-\delta_N)\Pr(E).
\end{align*}
For all large enough $N$, $\delta_N<1/2$, so
\begin{equation}\label{eq:early-crossings}
    \Pr\{E,S_N\le b_N\}
    \le\frac{1/2+\delta_N}{1/2-\delta_N}\Pr\{E,S_N>b_N\}
    \le\frac{1/2+\delta_N}{1/2-\delta_N}\Pr\{S_N>b_N\}.
\end{equation}
Using $T\wedge m$ keeps $\tau$ within the range covered by Assumption~3(a) even on paths that cross only late or never cross.

For a late disappearing crossing, $m<T\le N$ and $S_N\le b_N$. At the crossing time, $S_T>b_N$ and $|S_N-S_T|\le H_N$, implying
\begin{equation}\label{eq:late-crossings}
    \Pr\{m<T\le N,S_N\le b_N\}
    \le\Pr\{b_N-H_N<S_N\le b_N\}.
\end{equation}
By Assumption~3(b), there is a deterministic sequence $a_N\downarrow0$ for which $\Pr\{H_N>a_N\sigma_N\}\to0$. Therefore,
\begin{align*}
    \Pr\{b_N-H_N<S_N\le b_N\}
    &\le\Pr\{\gamma-a_N<S_N/\sigma_N\le\gamma\}
       +\Pr\{H_N>a_N\sigma_N\}\\
    &=o(1).
\end{align*}
To justify the last equality under the centred CLT in Assumption~2, let $Z_N=(S_N-\E[S_N])/\sigma_N$ and let $G_N$ be its CDF. Since the limiting CDF is continuous, $e_N:=\sup_x|G_N(x)-\Phi(x)|\to0$. Setting $c_N=\gamma-\E[S_N]/\sigma_N$, we obtain
\[
    \Pr\{\gamma-a_N<S_N/\sigma_N\le\gamma\}
    =G_N(c_N)-G_N(c_N-a_N)
    \le\frac{a_N}{\sqrt{2\pi}}+2e_N\to0.
\]

Combining \eqref{eq:crossing-decomposition}--\eqref{eq:late-crossings} gives
\begin{align*}
    \Pr\{M_N>b_N\}
    &\le\left(1+\frac{1/2+\delta_N}{1/2-\delta_N}\right)\Pr\{S_N>b_N\}+o(1)\\
    &=2\Pr\{S_N>b_N\}+o(1),
\end{align*}
since $\delta_N\to0$. This proves the lemma.
\end{proof}

\begin{proof}[Proof of Theorem~\ref{thm:sst}.]
Lemma~\ref{lem:levy} and the terminal CLT give, for fixed $\gamma>0$,
\begin{align*}
    \Pr\{M_N>\gamma\sqrt{NV_N}\}
    &\le2\Pr\{S_N/\sigma_N>\gamma\}+o(1)\\
    &=2(1-\Phi(\gamma))+o(1).
\end{align*}
Taking $\gamma=z_{1-\alpha/2}$ gives $2(1-\Phi(\gamma))=\alpha$ and proves the stated false detection bound.
\end{proof}

\begin{proof}[Proof of Corollary~\ref{cor:plugin}.]
Write $z=z_{1-\alpha/2}>0$. For any fixed $\eta\in(0,1)$,
\begin{align*}
    \Pr\{M_N>z\sqrt{N\hat V_N}\}
    &\le\Pr\{M_N>z(1-\eta)\sigma_N\}\\
    &\quad+\Pr\{\sqrt{N\hat V_N}<(1-\eta)\sigma_N\}.
\end{align*}
The second term tends to zero by \eqref{eq:variance-consistency}. Applying Theorem~\ref{thm:sst} at the fixed value $\gamma=z(1-\eta)$ yields
\[
    \limsup_{N\to\infty}\Pr\{M_N>z\sqrt{N\hat V_N}\}
    \le2(1-\Phi(z(1-\eta))).
\]
Letting $\eta\downarrow0$ gives $\limsup_{N\to\infty}p_N\le\alpha$, where $p_N=\Pr\{M_N>z\sqrt{N\hat V_N}\}$. Hence $e_N=\max\{p_N-\alpha,0\}\to0$ and $p_N\le\alpha+e_N$, proving \eqref{eq:plugin-control}.
\end{proof}

\begin{proof}[Proof of Corollary~\ref{cor:two-sided}.]
Apply Lemma~\ref{lem:levy} to $S_n$ and $-S_n$ and use the union bound:
\begin{align*}
    \Pr\left\{\max_{n\le N}|S_n|>\gamma\sigma_N\right\}
    &\le2\Pr\{S_N>\gamma\sigma_N\}
       +2\Pr\{-S_N>\gamma\sigma_N\}+o(1)\\
    &=2\Pr\{|S_N|>\gamma\sigma_N\}+o(1)\\
    &=4(1-\Phi(\gamma))+o(1).
\end{align*}
The terminal events are disjoint because $\gamma>0$, and the last equality follows from Assumption~2 and the zero mean under the strong null. 
Setting $\gamma=z_{1-\alpha/4}$ proves the oracle claim. 
For the plug-in claim, one can repeat the sandwich argument of 
Corollary~\ref{cor:plugin} with $\max_{n\le N}|S_n|$ 
and $z=z_{1-\alpha/4}$, applying the bound just proved at $z(1-\eta)$ and letting $\eta\downarrow0$.
\end{proof}

\begin{proof}[Proof of Proposition~\ref{prop:balance}.]
Fix $k\le N-L_N$. Conditional on $\mathcal G_{k,N}$, the event stream through event $N$ and the signs of seen users are fixed. Assumption~1 and \eqref{eq:H0a} leave the signs of unseen users independent and fair. Hence $R_{k,N}$ is a weighted sum of independent fair signs and is symmetric about zero. In particular,
\[
    \Pr\{R_{k,N}\ge0\mid\mathcal G_{k,N}\}\ge1/2.
\]
By \eqref{eq:balance-sufficient}, $B_{k,N}\ge-\epsilon_Ns_{k,N}$, so
\begin{align*}
    \Pr\{B_{k,N}+R_{k,N}\ge0\mid\mathcal G_{k,N}\}
    &\ge\Pr\{R_{k,N}\ge\epsilon_Ns_{k,N}\mid\mathcal G_{k,N}\}\\
    &\ge1/2-\Pr\{|R_{k,N}|\le\epsilon_Ns_{k,N}\mid\mathcal G_{k,N}\}\\
    &\ge1/2-\rho_N.
\end{align*}
Taking conditional expectations given $\mathcal F_k$ gives the same lower bound at time $k$. The hypotheses hold simultaneously at all such $k$, so this proves Assumption~3(a). Replacing $(B_{k,N},R_{k,N})$ by $(-B_{k,N},-R_{k,N})$ proves the other sign.
\end{proof}

\begin{proof}[Proof of Theorem~\ref{thm:power}.]
Let $\xi_N=(S_N-N\mu)/\sqrt{var(S_N)}$. By the assumed CLT, $\xi_N=O_p(1)$. Since the deterministic sequence $V_N=var(S_N)/N$ is bounded above,
\begin{equation}\label{eq:power-standardisation}
    \frac{S_N}{N}=\mu+\xi_N\sqrt{\frac{V_N}{N}}\to_p\mu>0.
\end{equation}
For fixed $\alpha\in(0,1)$, write $z=z_{1-\alpha/2}>0$. The condition $\hat V_N=O_p(1)$ under the alternative gives
\[
    \frac{z\sqrt{N\hat V_N}}{N}
    =z\sqrt{\frac{\hat V_N}{N}}\to_p0.
\]
Consequently,
\begin{align*}
    \Pr\left\{\max_{n\le N}S_n>z\sqrt{N\hat V_N}\right\}
    &\ge\Pr\left\{S_N>z\sqrt{N\hat V_N}\right\}\\
    &=\Pr\left\{\frac{S_N}{N}-z\sqrt{\frac{\hat V_N}{N}}>0\right\}\to1.
\end{align*}
\end{proof}

\section{Experiments using synthetic data}\label{sec:synthetic}

Our simulation setup follows a recent study~\cite{SSe:23}. We extended its implementation \cite{SSe:23:code} with YEAST, proposed in Section~\ref{sec:sst}.

As in~\cite{SSe:23:code}, we generate $N$ observations from the control, $Y_i^c$, $i=1,\,\ldots,\,N$, and $N$ observations from treatment, $Y_i^t$, with $N$ set to 500. Observations are independent and identically distributed within each arm, and the two arms are independent. The effect size (on the mean) took values 0.0, 0.1, 0.2, 0.3, and 0.4. In this appendix, we follow the simulation code's treatment-minus-control convention and set $X_i=Y_i^t-Y_i^c$, $i=1,\,\ldots,\,N$, so positive effects lie in the monitored upper tail. Here $N$ counts pairs of observations. This reverses the sign convention in the main text but leaves the null crossing analysis unchanged. Under the null, these paired increments are independent and symmetric, so the future sign-balance condition holds exactly.

We conducted two simulation experiments. The first experiment repeats the simulation study from~\cite{SSe:23} but adds YEAST to the comparison.
In the second set of simulations, we explored the behaviour of the compared methods when the increments are non-normal.

In all experiments, the target significance level was set at 5\% and the number of replications was set to 100,000.

The simulation code for these experiments can be found in the associated git repository \cite{KGu:24}.

\subsection{Simulation Results}\label{subsec:main}

In this experiment, observations $Y_i^c$ and $Y_i^t$ were drawn from normal distributions with parameters $(1,\,1)$ and $(1 + \xi,\,1)$, respectively. The effect size $\xi$ took values 0.0, 0.1, 0.2, 0.3, and 0.4. The random generator seed was set to 8163 (as in~\cite{SSe:23, SSe:23:code}).

In the following we present the list of methods we compared.

\noindent\textbf{YEAST} The proposed sequential method.

\noindent\textbf{YEASTn\{K\}} (with $K=80,\,90,\,110,\,120$) are the instances of the proposed method with the product $NV_N$ misestimated by a factor of $0.8,\,0.9,\,1.1$ and $1.2$, respectively (i.\,e. the method is applied with $NV_N$ set at 10\% and 20\% below and above the true value). Since the alerting boundary of YEAST depends on the product of $N$ and $V_N$ we can study the effect of inaccuracies in their estimation together.

\noindent\textbf{mSPRT} The mixture sequential probability ratio test \cite{Robins1970,Lindon22}. We set the tuning parameter $\phi$ to $1/0.3^2\approx11.11$, 25, and 100 and denote the corresponding versions as mSPRTphi11, mSPRTphi25, and mSPRTphi100.

\noindent\textbf{GAVI} The generalisation of the always valid inference, as proposed in~\cite{Howard2021}. As in~\cite{SSe:23}, we set the numerator of parameter $\rho$ of the method to 250, 500, and 750 and denote the corresponding instances of the method by GAVI250, GAVI500, and GAVI750, respectively.

\noindent\textbf{LanDeMetsOBF} The GST method \cite{GDe:83} with the O'Brien-Fleming alpha-spending function \cite{BFl:79}. For computational reasons, we constructed the boundary using 100 interim checkpoints and held the standardised critical values constant between checkpoints. This interpolation permits checking after every observation but does not inherit the discrete-look Type-I error guarantee.

\noindent\textbf{SeqC2ST-QDA} The sequential predictive test from \cite{PRa:23} with the online Newton step (ONS) strategy for selecting betting fractions. We used QDA (quadratic discriminant analysis) as a classification model for this method.

\noindent\textbf{Bonferroni} A naive approach using Bonferroni corrections.

All of the compared methods were employed in the continuous monitoring mode meaning that the check for significance was performed after each observation. 
In Appendix~\ref{app:extra} 
we report additional evaluation results for the case where YEAST was employed in the ``discrete mode'' (i.\,e. with a fixed number of interim significance checks).

For each experimental setting, we conducted 100,000 replications. Each replication can result in a detection or no-detection. A detection occurs when the respective test flags significance (at any point of the monitoring process). We compute the share of replications where a detection occurs. When the treatment does not have an effect this share is the so-called (empirical) false detection rate (or, synonymously, false positive rate, type-I error, or test ``size''). In settings where the treatment has an effect, this share is the (empirical) power. Table~\ref{tbl:mainExperiment} presents the measured false detection rate and power, along with normal-approximation confidence intervals, reported as estimate ± margin of error and adjusted by Bonferroni for nominal simultaneous 95\% coverage over the displayed method--effect cells. 

A method is excluded from the bold comparison when its Bonferroni-adjusted null confidence interval lies entirely above 5\%. Among the remaining methods, bold entries have intervals overlapping that of the highest-power method in the column.

One can see from the table that YEAST demonstrated the highest power among the continuous monitoring approaches that did not inflate the false detection rate. It means that while it kept the false positive rate under control when there was no treatment effect, YEAST had the highest sensitivity among the compared approaches when the treatment had an effect on the metric of interest.

\begin{table}[ht]
\caption{Simulation Experiment: \\ Empirical Type-I Error and Power}\label{tbl:mainExperiment}
\centering
\begin{small}
\begin{tabular}{rlrrrr}
  \hline
  & effect size & 0.0 & 0.1 & 0.2 & 0.3\\
  & method & & & & \\
  \hline
  1 & YEAST & 0.047 ± 0.002 & \textbf{0.439 ± 0.005} & \textbf{0.918 ± 0.003} & \textbf{0.998 ± 0.000} \\
  2 & YEASTn110 & 0.037 ± 0.002 & 0.398 ± 0.005 & 0.901 ± 0.003 & \textbf{0.998 ± 0.000} \\
  3 & YEASTn120 & 0.030 ± 0.002 & 0.362 ± 0.005 & 0.884 ± 0.003 & 0.997 ± 0.001 \\
  4 & YEASTn80 & 0.074 ± 0.003 & 0.530 ± 0.005 & 0.947 ± 0.002 & 0.999 ± 0.000 \\
  5 & YEASTn90 & 0.059 ± 0.002 & 0.483 ± 0.005 & 0.933 ± 0.003 & 0.999 ± 0.000 \\
  6 & mSPRTphi100 & 0.016 ± 0.001 & 0.216 ± 0.004 & 0.756 ± 0.005 & 0.987 ± 0.001 \\
  7 & mSPRTphi11 & 0.030 ± 0.002 & 0.220 ± 0.004 & 0.732 ± 0.005 & 0.982 ± 0.001 \\
  8 & mSPRTphi25 & 0.027 ± 0.002 & 0.235 ± 0.004 & 0.757 ± 0.005 & 0.986 ± 0.001 \\
  9 & GAVI250 & 0.024 ± 0.002 & 0.237 ± 0.004 & 0.765 ± 0.004 & 0.987 ± 0.001 \\
  10 & GAVI500 & 0.018 ± 0.001 & 0.226 ± 0.004 & 0.762 ± 0.004 & 0.987 ± 0.001 \\
  11 & GAVI750 & 0.014 ± 0.001 & 0.207 ± 0.004 & 0.748 ± 0.005 & 0.986 ± 0.001 \\
  12 & Bonferroni & 0.003 ± 0.001 & 0.044 ± 0.002 & 0.398 ± 0.005 & 0.896 ± 0.003 \\
  13 & LanDeMetsOBF & 0.054 ± 0.002 & 0.466 ± 0.005 & 0.927 ± 0.003 & 0.999 ± 0.000 \\
  14 & SeqC2ST-QDA & 0.022 ± 0.002 & 0.047 ± 0.002 & 0.216 ± 0.004 & 0.630 ± 0.005 \\
   \hline
\end{tabular}
\end{small}
\end{table}

In Section~\ref{app:savings} 
we additionally report sample savings that each of the methods produced on average (due to the early effect detection in an interim check).

\subsection{Non-Normal Data}

The derivation of YEAST involves the application of the Central Limit Theorem to sums of $X_i$ (increments in the metric difference between the two groups). For a fixed sample size, the quality of the normal approximation depends on the distribution of $X_i$. In this section, we explore the performance of YEAST in situations where the distribution of $Y_i^c$ and $Y_i^t$ is not normal. Specifically, we used two alternative distributions: Student's t which has heavier tails than the normal distribution and Gamma which is asymmetric. Student's t distribution had 3 degrees of freedom and was shifted by $\sqrt{3}$ for the control and by $\sqrt{3}(1+\xi)$ for the treatment. The shifting was done to maintain the same coefficient of variation as in Section~\ref{subsec:main}. The Gamma distribution had its shape parameter set to 1.0. The scale parameter equalled $2$ for the control and $2(1+\xi)$ for the treatment. The effect size $\xi$ took the same values as in the first simulation experiment: $\xi=0.0,\,0.1,\,0.2,\,0.3,\,0.4$. The random seed was set to 2023 for the simulations with the Student's t distribution and to 2024 for the simulations with the Gamma distribution. The two seeds were different to avoid dependence across the two sets of simulations.

The results are depicted in Figure~\ref{fig:power}. The first data point on each line corresponds to the case of no treatment effect and therefore represents the false detection rate. The remaining points represent the power for different treatment effect sizes. 
Similarly to the experiment with normal data, YEAST had a considerably higher power curve (both for the Gamma and the Student's t distribution cases) than other methods that did not inflate the false detection rate. 

\begin{figure}
\caption{Power Curves (under non-normal increment distributions)}\label{fig:power}
    \centering
    \includegraphics[width=\textwidth]{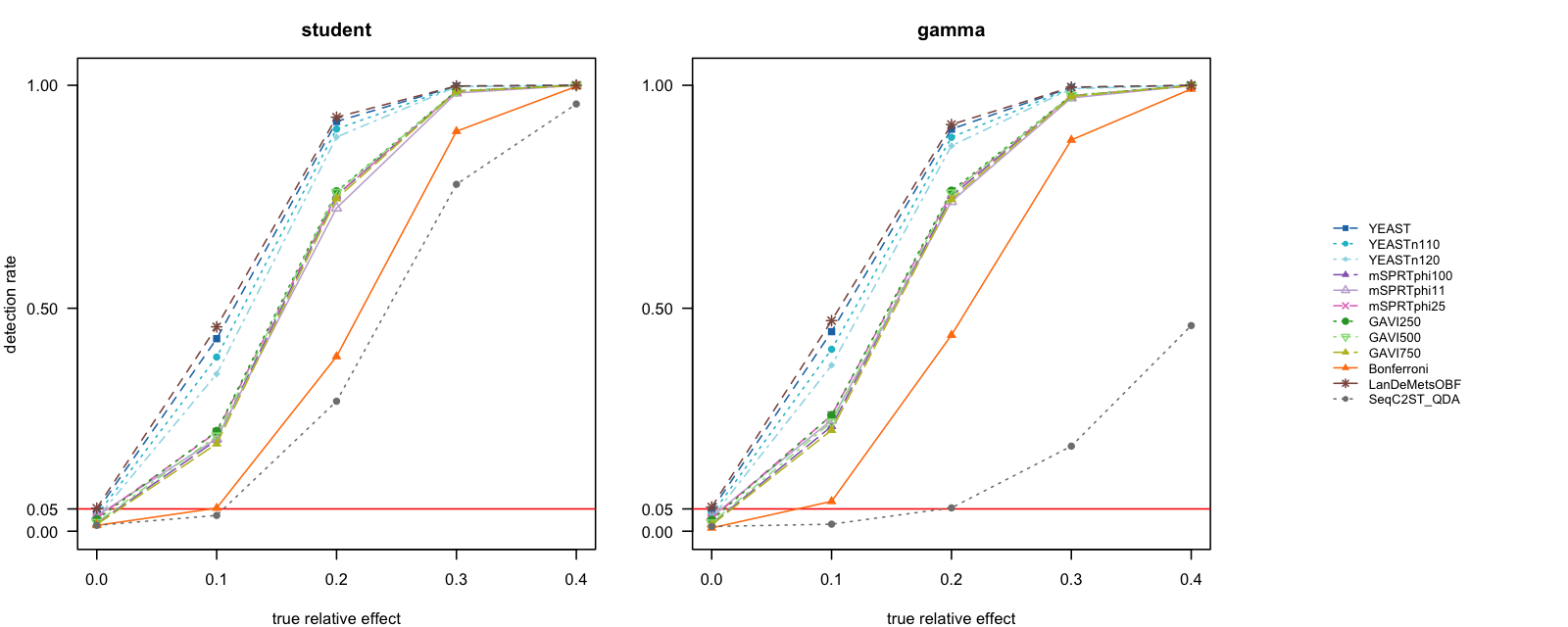}
\end{figure}

\section{Discrete Monitoring}\label{app:extra}

In this section, we report additional evaluation results for the case where YEAST was employed in the discrete mode (i.\,e. with only a fixed number of interim checks). The evaluation was performed against the same replications as in Section~\ref{subsec:main}. We again follow the protocol from~\cite{SSe:23} and perform 14, 28, 42, and 56 significance checks (spaced equally across the timeline). Since in these evaluations, we operate in a discrete setting, we were able to include the discrete baselines from~\cite{SSe:23} in the comparison. Namely, we compare against the following benchmark.

\noindent\textbf{GST} The group sequential test with alpha spending as proposed in~\cite{GDe:83}. The test performs a significance check after the arrival of each batch of observations. To schedule a prespecified number of checks it therefore needs an estimate of the total number of observations that would be collected during the experiment time frame. In the evaluations, we consider three different scenarios: when the sample size is estimated precisely (right number of checks), when the sample size is underestimated (leading to more checks than planned), and when the sample size is overestimated (leading to making fewer checks than planned). The respective entries in the evaluation table are referred to as \emph{GST}, \emph{GSToversampled}, and \emph{GSTundersampled}, respectively. The actual sample size was 500 and the respective sample size estimates for the three scenarios were 500, 250, and 1,000. In the case of oversampling (i.\,e. the sample size is underestimated), we apply the correction to the bounds proposed in~\cite[pp. 78--79]{WBr:16}. We consider linear and cubic alpha-spending, the latter having the ``phi3'' suffix in the name.

We report the share of replications where the test detects an effect (i.\,e. the significance is detected in at least one of the interim checks). Table~\ref{tbl:appendixFPR} (``False Positive Rate'') reports this share for the case where no actual effect was present. All of the compared methods except oversampled versions of GST keep the false detection rate below the nominal level of 5\%. Table~\ref{tbl:appendixPower} (``Power'') contains the share of replications with a detection for the case where the treatment effect was set to 0.2 standard deviations. The methods with an inflated false detection rate were excluded from the power comparisons. The GST method with cubic alpha-spending demonstrated the highest power, closely followed by YEAST and GST with linear alpha-spending. We find it remarkable that our proposed method, despite supporting continuous monitoring, performed on par with the GST method in the discrete monitoring setting.

\begin{table}[ht]
\caption{False Positive Rate}\label{tbl:appendixFPR}
\centering
\begin{small}
\begin{tabular}{rlrrrr}
  \hline
  & type & 14 & 28 & 42 & 56 \\
  \hline
  1 & YEAST & 0.04 & 0.04 & 0.04 & 0.04 \\
  2 & GST & 0.05 & 0.05 & 0.05 & 0.05 \\
  3 & GSTphi3 & 0.05 & 0.05 & 0.05 & 0.05 \\
  4 & GSToversampled & 0.07 & 0.07 & 0.07 & 0.08 \\
  5 & GSToversampledphi3 & 0.09 & 0.10 & 0.10 & 0.10 \\
  6 & GSTundersampled & 0.03 & 0.02 & 0.03 & 0.03 \\
  7 & GSTundersampledphi3 & 0.01 & 0.01 & 0.01 & 0.01 \\
   \hline
\end{tabular}
\end{small}
\end{table}

\begin{table}[ht]
\caption{Power (under a treatment effect of 0.2 standard deviations)}\label{tbl:appendixPower}
\centering
\begin{small}
\begin{tabular}{rlrrrrr}
  \hline
  & type & 14 & 28 & 42 & 56 & stream \\
  \hline
  1 & YEAST & 0.90 & 0.91 & 0.91 & 0.91 & 0.92 \\
  2 & GST & 0.90 & 0.90 & 0.90 & 0.89 & - \\
  3 & GSTphi3 & 0.93 & 0.93 & 0.93 & 0.93 & - \\
  4 & GSTundersampled & 0.83 & 0.83 & 0.82 & 0.82 & - \\
  5 & GSTundersampledphi3 & 0.72 & 0.72 & 0.71 & 0.71 & - \\
   \hline
\end{tabular}
\end{small}
\end{table}

\section{Sample/Time Savings}\label{app:savings}

The main benefit of sequential testing is the ability to stop the experiment early (once significance is flagged upon one of the interim checks). This saves observations and, if the treatment is harmful, reduces the negative impact. Thus, the amount of savings that is generated by a sequential test on average is another important metric and we report the savings observed in the experiment from Section~\ref{subsec:main} in Table~\ref{tbl:mainExperimentSavings}. The savings are measured as follows: if a method identified the effect after the arrival of 10\% of the total number of observations, the associated sample savings would be 90\%. Replications without a detection contribute zero, and the average is taken over all replications. Sample savings approximate calendar-time savings only when observations arrive at an approximately constant rate.

\begin{table}[ht]
\caption{Simulation Experiment: Sample Savings, \%}\label{tbl:mainExperimentSavings}
\centering
\begin{small}
\begin{tabular}{rlrrrr}
  \hline
  & effect size & 0.1 & 0.2 & 0.3 & 0.4\\
  & method & & & & \\
  \hline
  1 & YEAST & 13 & 39 & 58 & 69 \\
  2 & mSPRTphi100 & 9 & 35 & 62 & 75 \\
  3 & mSPRTphi11 & 12 & 41 & 68 & 82 \\
  4 & mSPRTphi25 & 12 & 41 & 68 & 81 \\
  5 & GAVI250 & 12 & 40 & 67 & 80 \\
  6 & GAVI500 & 10 & 37 & 63 & 77 \\
  7 & GAVI750 & 8 & 34 & 60 & 74 \\
  8 & Bonferroni & 2 & 15 & 45 & 67 \\
  9 & LanDeMetsOBF & 14 & 41 & 59 & 70 \\
  10 & SeqC2ST-QDA & 3 & 11 & 30 & 54 \\
   \hline
\end{tabular}
\end{small}
\end{table}


\begin{thebibliography}{10}

\bibitem{Cam:11}
A.~Colin Cameron, Jonah~B. Gelbach, and Douglas~L. Miller.
\newblock Robust inference with multiway clustering.
\newblock {\em Journal of Business \& Economic Statistics}, 29(2):238--249, 2011.

\bibitem{armitage1954sequential}
Peter Armitage.
\newblock Sequential tests in prophylactic and therapeutic trials.
\newblock {\em Quarterly Journal of Medicine}, 23(91):255--74, 1954.

\bibitem{Bet:98}
Rebecca~A. Betensky.
\newblock Construction of a continuous stopping boundary from an alpha spending function.
\newblock {\em Biometrics}, 54(3):1061--1071, 1998.
\newblock \url{https://doi.org/10.2307/2533857}.

\bibitem{BKa:22}
Aur{\'e}lien~F. Bibaut, Nathan Kallus, and Michael Lindon.
\newblock Near-optimal non-parametric sequential tests and confidence sequences with possibly dependent observations.
\newblock 2022.

\bibitem{bross1952sequential}
I~Bross.
\newblock Sequential medical plans.
\newblock {\em Biometrics}, 8(3):188--205, 1952.

\bibitem{Che:15}
Daqing Chen.
\newblock {Online Retail}.
\newblock UCI Machine Learning Repository, 2015.
\newblock {DOI}: \url{https://doi.org/10.24432/C5BW33}.

\bibitem{dodge1929method}
Harold~French Dodge and Harry~G Romig.
\newblock A method of sampling inspection.
\newblock {\em The Bell System Technical Journal}, 8(4):613--631, 1929.

\bibitem{doob1953stochastic}
J.L. Doob.
\newblock {\em Stochastic Processes}.
\newblock Probability and Statistics Series. Wiley, 1953.

\bibitem{GDe:83}
K.~K. Gordon~Lan and David~L. DeMets.
\newblock {Discrete sequential boundaries for clinical trials}.
\newblock {\em Biometrika}, 70(3):659--663, 12 1983.

\bibitem{Grunwald_deHeide_Koolen_2024_SafeTesting}
Peter Grünwald, Rianne de~Heide, and Wouter~M. Koolen.
\newblock Safe testing.
\newblock {\em Journal of the Royal Statistical Society. Series B: Statistical Methodology}, 86(5):1091--1128, 2024.

\bibitem{ham2023designbased}
Dae~Woong Ham, Iavor Bojinov, Michael Lindon, and Martin Tingley.
\newblock Design-based confidence sequences: A general approach to risk mitigation in online experimentation, 2023.

\bibitem{Howard2021}
Steven~R. Howard, Aaditya Ramdas, Jon McAuliffe, and Jasjeet Sekhon.
\newblock Time-uniform, nonparametric, nonasymptotic confidence sequences.
\newblock {\em The Annals of Statistics}, 49(2):1055--1080, 2021.

\bibitem{Bandit4}
Kevin Jamieson and Lalit Jain.
\newblock A bandit approach to multiple testing with false discovery control.
\newblock In {\em Proceedings of the 32nd International Conference on Neural Information Processing Systems}, NIPS'18, page 3664–3674, Red Hook, NY, USA, 2018. Curran Associates Inc.

\bibitem{Bandit1}
Kevin Jamieson, Matthew Malloy, Robert Nowak, and Sébastien Bubeck.
\newblock lil' UCB: An optimal exploration algorithm for multi-armed bandits, 2013.

\bibitem{johari2022always}
Ramesh Johari, Pete Koomen, Leonid Pekelis, and David Walsh.
\newblock Always valid inference: Continuous monitoring of a/b tests.
\newblock {\em Operations Research}, 70(3):1806--1821, 2022.

\bibitem{Bandit2}
Emilie Kaufmann, Olivier Capp{\'e}, and Aur{\'e}lien Garivier.
\newblock On the complexity of best arm identification in multi-armed bandit models.
\newblock {\em Journal of Machine Learning Research}, 17:1--42, 2016.

\bibitem{KGu:24}
Alexey Kurennoy.
\newblock Yeast: Yet another sequential test.
\newblock Git Repository, 2024.
\newblock This paper's code companion. \url{https://github.com/akurennoy/yeast}.

\bibitem{larsen2023statistical}
Nicholas Larsen, Jonathan Stallrich, Srijan Sengupta, Alex Deng, Ron Kohavi, and Nathaniel Stevens.
\newblock Statistical challenges in online controlled experiments: A review of a/b testing methodology, 2023.

\bibitem{LHa:22}
Michael Lindon, Dae~Woong Ham, Martin~P. Tingley, and Iavor Bojinov.
\newblock Anytime-valid linear models and regression adjusted causal inference in randomized experiments.
\newblock 2022.

\bibitem{LAl:22}
Michael Lindon and Alan Malek.
\newblock Anytime-valid inference for multinomial count data.
\newblock In S.~Koyejo, S.~Mohamed, A.~Agarwal, D.~Belgrave, K.~Cho, and A.~Oh, editors, {\em Advances in Neural Information Processing Systems}, volume~35, pages 2817--2831. Curran Associates, Inc., 2022.

\bibitem{Lindon22}
Michael Lindon, Chris Sanden, and Vach\'{e} Shirikian.
\newblock Rapid regression detection in software deployments through sequential testing.
\newblock In {\em Proceedings of the 28th ACM SIGKDD Conference on Knowledge Discovery and Data Mining}, KDD '22, page 3336–3346, New York, NY, USA, 2022. Association for Computing Machinery.

\bibitem{Loe:60}
M.~Lo\`eve.
\newblock {\em Probability Theory. 2nd Edition}.
\newblock D. Van Nostrand Company Inc., 1960.

\bibitem{Mi:15}
Evan Miller.
\newblock Simple sequential A/B testing.
\newblock Blog Post, 2015.
\newblock \url{https://www.evanmiller.org/sequential-ab-testing.html}

\bibitem{MPe:10}
Peter M\"orters and Yuval Peres.
\newblock {\em Brownian Motion}.
\newblock Cambridge Series in Statistical and Probabilistic Mathematics, volume~30. Cambridge University Press, 2010.
\newblock \url{https://doi.org/10.1017/CBO9780511750489}.

\bibitem{BFl:79}
Peter~C. O'Brien and Thomas~R. Fleming.
\newblock A multiple testing procedure for clinical trials.
\newblock {\em Biometrics}, 35(3):549--556, 1979.

\bibitem{Po:77}
Stuart~J. Pocock.
\newblock Group sequential methods in the design and analysis of clinical trials.
\newblock {\em Biometrika}, 64(2):191--199, 1977.

\bibitem{PRa:23}
Aleksandr Podkopaev and Aaditya Ramdas.
\newblock Sequential predictive two-sample and independence testing.
\newblock In A.~Oh, T.~Naumann, A.~Globerson, K.~Saenko, M.~Hardt, and S.~Levine, editors, {\em Advances in Neural Information Processing Systems}, volume~36, pages 53275--53307. Curran Associates, Inc., 2023.

\bibitem{Ramdas_Wang_2025_HypothesisTestingWithEvalues}
Aaditya Ramdas and Ruodu Wang.
\newblock {\em Hypothesis Testing with E-values}.
\newblock {\em Foundations and Trends in Statistics}, 1(1--2):1--390, 2025.
\newblock \url{https://doi.org/10.1561/3600000002}.

\bibitem{Robins1970}
Herbert Robbins.
\newblock {Statistical Methods Related to the Law of the Iterated Logarithm}.
\newblock {\em The Annals of Mathematical Statistics}, 41(5):1397 -- 1409, 1970.

\bibitem{SSe:23}
Mårten Schultzberg and Sebastian Ankargren.
\newblock Choosing a sequential testing framework — comparisons and discussions.
\newblock Blog Post, 2023.
\newblock \url{https://engineering.atspotify.com/2023/3/choosing-sequential-testing-framework-comparisons-and-discussions}.

\bibitem{SSe:23:code}
Mårten Schultzberg and Sebastian Ankargren.
\newblock Simulation files for Schultzberg \& Ankargren's blogpost.
\newblock GitHub repository, 2023.
\newblock \url{https://github.com/MSchultzberg/SeqTestBlog}

\bibitem{Shafer_2021_TestingByBetting}
Glenn Shafer.
\newblock Testing by betting: A strategy for statistical and scientific communication.
\newblock {\em Journal of the Royal Statistical Society. Series A: Statistics in Society}, 184(2):407--431, 2021.

\bibitem{Shekhar_Ramdas_2021_NonparametricTwoSampleByBetting_arxiv}
Shubhanshu Shekhar and Aaditya Ramdas.
\newblock Game-theoretic formulations of sequential nonparametric one- and two-sample tests.
\newblock {\em arXiv preprint arXiv:2112.09162}, 2021.
\newblock Later version published in IEEE Transactions on Information Theory (2024).

\bibitem{Shekhar_Ramdas_2023_NonparametricTwoSampleByBetting_TIT}
Shubhanshu Shekhar and Aaditya Ramdas.
\newblock Nonparametric two-sample testing by betting.
\newblock {\em IEEE Transactions on Information Theory}, 70(2):1178--1203, 2024.

\bibitem{SWe:82}
Eric~V. Slud and L.~J. Wei.
\newblock Two-sample repeated significance tests based on the modified wilcoxon statistic.
\newblock {\em Journal of the American Statistical Association}, 77:862--868, 1982.

\bibitem{Va:00}
A.W. van~der Vaart.
\newblock {\em Asymptotic Statistics}.
\newblock Asymptotic Statistics. Cambridge University Press, 2000.

\bibitem{Vovk_1993_LogicOfProbability}
Vladimir Vovk.
\newblock A logic of probability with application to the foundations of statistics.
\newblock {\em Journal of the Royal Statistical Society. Series B (Methodological)}, 55(2):317--341, 1993.

\bibitem{Wald45}
A.~Wald.
\newblock {Sequential Tests of Statistical Hypotheses}.
\newblock {\em The Annals of Mathematical Statistics}, 16(2):117 -- 186, 1945.

\bibitem{WBr:16}
G.~Wassmer and W.~Brannath.
\newblock {\em Group Sequential and Confirmatory Adaptive Designs in Clinical Trials}.
\newblock Springer Series in Pharmaceutical Statistics. Springer International Publishing, 2016.

\bibitem{Whi:84}
H.~White.
\newblock {\em Asymptotic Theory for Econometricians}.
\newblock Economic Theory, Econometrics, and Mathematical Economics. Elsevier Science, 1984.

\bibitem{Bandit3}
Fanny Yang, Aaditya Ramdas, Kevin Jamieson, and Martin Wainwright.
\newblock A framework for multi-a(rmed)/b(andit) testing with online fdr control.
\newblock In {\em Proceedings of the 31st International Conference on Neural Information Processing Systems}, NIPS'17, page 5959–5968, Red Hook, NY, USA, 2017. Curran Associates Inc.

\bibitem{Zei:04}
Achim Zeileis.
\newblock Econometric computing with {HC} and {HAC} covariance matrix estimators.
\newblock {\em Journal of Statistical Software}, 11(10):1--17, 2004.

\bibitem{ZSu:20}
Achim Zeileis, Susanne K\"oll, and Nathaniel Graham.
\newblock Various versatile variances: An object-oriented implementation of clustered covariances in {R}.
\newblock {\em Journal of Statistical Software}, 95(1):1--36, 2020.

\end{thebibliography}
\end{document}